\documentclass[a4paper,UKenglish,cleveref, autoref,thm-restate]{lipics-v2021}

\title{Localising Stochasticity in Weighted Automata}

\author{Smayan Agarwal}{Ashoka University, India}{smayan.agarwal@alumni.ashoka.edu.in}{}{}
\author{Aalok Thakkar}{Vachani School of Advanced Computing, Ashoka University, India}{thakkar@ashoka.edu.in}{0000-0002-3195-585X}{}
\authorrunning{S. Agarwal and A. Thakkar}
\Copyright{Smayan Agarwal and Aalok Thakkar}
\funding{This work was supported by the Anusandhan National Research
Foundation, under the Prime Minister Early Career
Research Grant ANRF/ECRG/2025/001136/ENS, and by the Vachani School of
Advanced Computing, Ashoka University, under the Young Faculty Grant
(00281).}



\usepackage{stmaryrd}
\usepackage{tikz}
\usetikzlibrary{automata, positioning, shapes.multipart}
\usepackage{algorithm, algpseudocode}

\newcommand{\semantics}[1]{{\llbracket #1 \rrbracket}}

\newcommand{\R}{\mathbb{R}}

\newcommand{\Rp}{\mathbb{R}_{\geq 0}}

\newcommand{\SZero}{\mathbf{0}}
\newcommand{\SOne}{\mathbf{1}}
\EventEditors{Raghavan Komondoor, Ramprasad Saptharishi, and James Worrell}
\EventNoEds{3}
\EventLongTitle{46th IARCS Annual Conference on Foundations of Software Technology and Theoretical Computer Science (FSTTCS 2026)}
\EventShortTitle{FSTTCS 2026}
\EventAcronym{FSTTCS}
\EventYear{2026}
\EventDate{December 16--18, 2026}
\EventLocation{IIT Delhi, New Delhi, India}
\EventLogo{}
\SeriesVolume{400}
\ArticleNo{15}
\begin{document}
\nolinenumbers

\maketitle
\begin{abstract}
Probability distributions over strings can be expressed by \emph{locally stochastic} weighted automata: state machines in which, at every state, the outgoing and final weights sum to the multiplicative identity of the semiring, and the initial weights do too. This is a syntactic condition that makes sense over any semiring with enough structure, and over
$\mathbb{R}_{\geq0}$. it gives the usual probabilistic automata. We study it in two parts.

First, we ask which semiring axioms
make local stochasticity preserved by the constructors of
\emph{stochastic regular expressions} (SREs): atoms, weighted choice,
concatenation, and a discounted Kleene star. Second, we ask when a given automaton can be rescaled into
locally stochastic form. Over the nonnegative reals, we give an
effective normalisation for automata of finite mass using
Perron--Frobenius theory, a Kleene--Sch\"utzenberger theorem for SREs,
and a decomposition of every automaton into a growth rate, a scale, and
a locally stochastic part. Over the tropical and Viterbi semirings, the
analogous normalisations use cycle means and completion costs. In all
three cases the rescaling needs division or subtraction, which the
semiring axioms do not provide, and we supply it separately for each
semiring.
\end{abstract}

\section{Introduction}
\label{sec:intro}

Weighted automata assign a value from a semiring to every word. If the semiring is of nonnegative reals, these machines can describe
probability distributions over strings. The standard model for this is a probabilistic automaton
\cite{paz1971probabilistic,rabin1963probabilistic}, in which the
probability is spread over the states: at every state, the weights of
the outgoing transitions together with the weight of stopping sum to
one, so a string can be sampled one step at a time. Labelled Markov
chains are the same model without the stopping weights
\cite{vanglabbeek1995reactive}.

This is a \emph{local} condition on states, and is strictly stronger than asking if the weights of all
words sum to one. Clearly, this is not a necessary condition and an automaton can define a distribution without any of
its states satisfying the local property (see \Cref{fig:running-example}). In this paper, we first ask when a weighted automaton is, up to rescaling, a probabilistic
automaton, so that a property of the whole function becomes a property
of each state. We call the state-by-state condition \emph{local
stochasticity}, and we study it over the $\R_{\ge 0}$, Viterbi, and tropical semirings.

A second reason to look at the condition comes from regular
expressions. Giving a semantics to the Kleene star of a weighted
expression has traditionally forced a choice among three routes. One can
restrict attention to \emph{proper} (quasiregular) expressions, in which
the star is applied only to a subexpression with no $\varepsilon$-weight,
so that the relevant sum is finitely supported at every word. One can
instead require the semiring to be \emph{complete} or
\emph{$\omega$-continuous}, so that arbitrary countable sums exist and
the star is one of them. Or one can specify the star axiomatically, as
the least solution of its defining fixed-point equation, in the style of
Conway and Kleene algebra. Each route is well developed
(\Cref{sec:relwork}). We ask instead which part of the semiring's
structure the star needs, and whether this can be checked state by
state.

Our answer is a syntactic condition on weighted automata and on
stochastic regular expressions (SREs). An automaton is \emph{locally
stochastic} if, at each state, the outgoing and final weights sum to the
multiplicative identity, and the initial weights also sum to the
identity, which is checked once for the whole automaton. The discounted
star needs one more side-condition, \emph{guardedness}: the expression
under the star must not match the empty word. Local stochasticity can be
established one grammar rule at a time, and the rules need different
amounts of structure from the semiring $S$. This gives two questions.

\begin{enumerate}
\item[\textbf{Q1.}] \emph{Semiring level.} For which $S$ is local
stochasticity a meaningful invariant, and do the SRE constructors
preserve it? (\Cref{sec:rung0,sec:rung2}.)
\item[\textbf{Q2.}] \emph{Automaton level.} Can a given weighted
automaton over such an $S$ be rescaled into locally stochastic form?
(\Cref{sec:realisability,sec:realisability-trop}.)
\end{enumerate}

Q1 depends on algebraic properties of $S$: completeness, positivity, and
whether the closure equation of the star has a solution
(\Cref{sec:rung0,sec:rung2}). Q2 depends on whether $S$ has an operation
that completeness and positivity do not provide: division in
$\mathbb{R}_{\geq0}$ and the Viterbi semiring, and subtraction in the
tropical semiring.

\paragraph*{Contributions.}
Our primary contribution is a construction for semantically equivalent locally stochastic automata (\Cref{thm:conjugation}). More concretely:
\begin{enumerate}
\item We identify the properties of semiring $S$ to convert a given automata into a locally stochastic one, and we show they hold in case of 
$\mathbb{R}_{\geq0}$, the tropical semiring, and the Viterbi semiring.
\item For Q2 over $\mathbb{R}_{\geq0}$, we prove that every
finite-mass automaton can be rescaled by its mass into an equivalent
probabilistic automaton (\Cref{thm:finite-mass-normal-form}). We prove
a Kleene--Sch\"utzenberger correspondence between probabilistic
automata and SREs (\Cref{thm:kleene-main}). We also decompose any
automaton into a growth rate, a scale, and a stochastic part
(\Cref{thm:spectral-decomposition}).
\item We prove analogous results for the tropical semiring, with cycle
means in place of the growth rate and least completion costs in place
of the scale (\Cref{thm:tropical-decomposition,thm:tropical-local}).
We finally prove it for the Viterbi semiring by rescaling with the best
completion vector (\Cref{thm:viterbi-decomposition}). Unlike the other
two cases, this needs no spectral condition.
\end{enumerate}

\section{Related Work}
\label{sec:relwork}

Some weighted automata and probabilistic grammars define probability
distributions over words~\cite{droste2007weighted}. Stochasticity has
been studied extensively in this setting, but our definitions and aims
differ. We ask when a weighted automaton over the nonnegative reals
\emph{is}, up to rescaling, a locally stochastic probabilistic
automaton, and we make the normalisation effective.

The literature has two largely independent traditions, which differ in
where probability must sum to one. In the model of
Paz~\cite{paz1971probabilistic}, the outgoing weight of each state (over all
letters, together with the termination weight) sums to one, so a run induces
a distribution over \emph{words}; this is the \emph{generative} model of van
Glabbeek, Smolka and Steffen~\cite{vanglabbeek1995reactive}, coinciding with
labelled Markov chains. In the model of Rabin~\cite{rabin1963probabilistic},
the matrix of \emph{each letter} is stochastic, distributing probability over
successor states given a letter; this \emph{reactive} model resolves
nondeterminism before probabilistic branching. We adopt the generative (Paz) model, because some weighted automata
define probability distributions that no Rabin automaton expresses.
Every Rabin automaton, in turn, embeds in a Paz automaton by adding a
dead state. Chatterjee, Doyen and
Henzinger~\cite{chatterjee2009probabilistic} instead study probabilistic
automata on words under cut-off acceptance, where words are recognised
by threshold conditions and the automaton does not define a
distribution over words.

Our stochastic regular expressions and the accompanying Kleene-style characterisation theorem relate to two lines
of work. Bollig, Gastin, Monm\`ege and
Zeitoun~\cite{Bollig2012ProbabilisticKleene} establish a probabilistic version of this theorem 
for the Rabin model, with expressions built from probabilistic and
guarded choice, concatenation, and a star operator; because their automata
are Rabin rather than Paz, their expressions denote a different kind of
object from ours. More generally, Silva, Bonchi, Bonsangue and
Rutten~\cite{silva2011quantitative} give a coalgebraic recipe that
derives, for a wide range of quantitative system types (transitions
valued in a monoid modelling cost, duration, or probability), a language
of generalised regular expressions with a sound and complete
axiomatisation. Our approach is concrete. W1e give an explicit syntax with a geometrically weighted star
whose parameter is computable, and we connect it to the finite-mass
normalisation through Perron--Frobenius theory, which the coalgebraic
account does not address. Our equivalence
between probabilistic automata and SREs is a stochastic instance of the
classical Kleene--Sch\"utzenberger correspondence.

Two earlier works use the name stochastic regular expressions, but
neither overlaps with ours. Ross~\cite{Ross2000} uses them as a representation for
evolutionary pattern matching over strings, with no semiring semantics,
discounted star, or equivalence with a class of automata. Beeh~\cite{Beeh2017}
studies transformations between stochastic regular expressions and
discrete-time Markov chains. Such chains are locally stochastic by
construction, so this equivalence is the restriction of
\Cref{thm:kleene-main} to inputs already in locally stochastic form. It
does not contain the finite-mass rescaling of arbitrary weighted
automata (\Cref{thm:finite-mass-normal-form}), the semiring-parametric
formulation, the closure-equation reduction of the star, or the tropical
and Viterbi instances.

A distinct line of work, initiated by Denis and Esposito~\cite{Denis2009LearningWA,Denis2004}, studied rational stochastic languages through multiplicity automata. While their models provide an algebraic vehicle for representing stochastic languages, the literature still lacks a clear syntactic or compositional characterization analogous to the Kleene–Schützenberger theorem for rational stochastic languages.
Our finite-mass characterisation answers this
constructively over $\mathbb{R}_{\ge 0}$. Identity testing of black-box distribution against rational stochastic languages, and their use in modelling recurrent neural networks has also been considered ~\cite{agarwal2026transformers}. For a broader survey of
probabilistic finite-state machines, see Vidal et
al.~\cite{vidal2005statistical}.

Finally, the \emph{stochastization} of Avni and
Kupferman~\cite{avni2015stochastization} has a similar name but
addresses a different question. Given a weighted automaton under
min-of-runs semantics, they relabel transitions with probabilities to
obtain a probabilistic weighted automaton whose word values increase as little as possible, an approximation problem
measured by a multiplicative factor. Our normalisation is exact and
preserves semantics, and it yields a generative distribution instead of
approximating a min-plus value. The two results are unrelated.

\section{Stochastic Regular Expressions}
\label{sec:rung0}

Definitions of semirings, weighted automata, and the linear algebra we
use are given in \Cref{app:matrix-properties}. This section introduces
\emph{local stochasticity}, the syntax and semantics of \emph{stochastic
regular expressions} (SREs), and the invariant that connects them: every
SRE denotes a word function of total mass~$\SOne$.

\subsection{Local stochasticity}

\begin{definition}[Local stochasticity]
\label{def:local-stoch-general}
A weighted automaton $A=(\Sigma,Q,\lambda,M,\mu)$ over
$(S,+,\times,\SZero,\SOne)$ is \emph{locally stochastic} if
$\sum_{q\in Q}\lambda(q)=\SOne$ and, for every state $q\in Q$, the
outgoing weight plus the final weight is exactly $\SOne$:
\[
\mu(q)\ +\ \sum_{a\in\Sigma}\sum_{q'\in Q}M_a(q,q')\ =\ \SOne .
\]
Both sums range over the finite sets $\Sigma$ and $Q$, so stating the
condition needs only the semiring axioms.
\end{definition}

This generalises stochasticity for weighted automata over
$\mathbb{R}_{\geq0}$, and it is a finitary condition. To speak of the
\emph{mass} of a whole word function, or to state an invariant of the
SRE grammar below, we must sum over the possibly infinite set
$\Sigma^*$, which requires a complete semiring.

\subsection{Syntax}

Fix a complete semiring $(S,+,\times,\SZero,\SOne)$.

\begin{definition}[SRE syntax]
\label{def:sre-syntax}
The set of \emph{stochastic regular expressions} over $S$ is generated
by the grammar
\[
r ::= \delta_w \mid a_1 r_1 + a_2 r_2 \mid r_1\cdot r_2 \mid r^*_t,
\qquad w\in\Sigma^*,\ a_1,a_2,t\in S,
\]
where every subexpression under a star, $r^*_t$, is required to be
\emph{guarded} (\Cref{def:guarded}). This is the only syntactic
side-condition. The grammar places no restriction on the scalars
$a_1$, $a_2$, and $t$, or on the word $w$.
\end{definition}

\begin{definition}[Guardedness]
\label{def:guarded}
An SRE $r$ is \emph{guarded} if $\llbracket r\rrbracket(\varepsilon)=\SZero$,
i.e.\ its semantics (below) assigns no weight to the empty word.
\end{definition}

\subsection{Semantics}

Each SRE $r$ denotes a word function $\llbracket r\rrbracket:\Sigma^*\to S$,
defined by recursion on the grammar of \Cref{def:sre-syntax}:
\[
\begin{aligned}
\llbracket\delta_w\rrbracket(w') &= \SOne \text{ if } w=w', \text{ else } \SZero, \\
\llbracket a_1r_1+a_2r_2\rrbracket(w) &= a_1\,\llbracket r_1\rrbracket(w) + a_2\,\llbracket r_2\rrbracket(w), \\
\llbracket r_1\cdot r_2\rrbracket(w) &= \sum_{w=uv} \llbracket r_1\rrbracket(u)\,\llbracket r_2\rrbracket(v), \\
\llbracket r^*_t\rrbracket(w) &= \sum_{k\ge1}\sum_{w=w_1\cdots w_k}
\llbracket r\rrbracket(w_1)\,t\,\llbracket r\rrbracket(w_2)\,t\cdots t\,\llbracket r\rrbracket(w_k)\,e,
\quad w_i\in\Sigma^+,
\end{aligned}
\]
where in the last clause $e\in S$ is determined from $t$ by the
\emph{closure equation} $t^*e=\SOne$ (\Cref{def:closure-eq} below). We
write $r(w)$ for $\llbracket r\rrbracket(w)$ where no confusion can
arise.

We call $r$ \emph{stochastic} if $\|\llbracket r\rrbracket\|=\SOne$, and
say the choice weights $a_1,a_2$ in the second clause are
\emph{normalised} if $a_1+a_2=\SOne$.

\section{Positivity and the Star}
\label{sec:rung2}

It remains to show that every SRE built from stochastic parts with
appropriate weights is itself stochastic. This is immediate for atoms
($\delta_w$) and for normalised choice ($a_1r_1 + a_2r_2$). We prove it
for the Cauchy product and the discounted Kleene star.

\begin{lemma}[Idempotent $\Rightarrow$ positive]
\label{lem:idempotent-positive}
If $+$ is idempotent, then $S$ is positive.
\end{lemma}
\begin{proof}
Define $a\le b\iff a+b=b$. This is a partial order: reflexive by
idempotence; antisymmetric since $a+b=b$ and $b+a=a$ together with
commutativity of $+$ force $a=b$; transitive since $a\le b\le c$ gives
$a+c=a+(b+c)=(a+b)+c=b+c=c$. The identity $\SZero$ is the bottom
element, since $\SZero+a=a$ gives $\SZero\le a$ for every $a$. Also
$a\le a+b$ for any $a,b$, since $a+(a+b)=(a+a)+b=a+b$. Now suppose
$a+b=\SZero$. Then $a\le a+b=\SZero$, and we already have
$\SZero\le a$; by antisymmetry $a=\SZero$. Symmetrically $b=\SZero$.
The same argument extends to an arbitrary family, directly from
$a_i\le\sum_j a_j$.
\end{proof}

\begin{proposition}[Concatenation]
\label{prop:concatenation}
If $r_1,r_2$ are stochastic, then $r_1\cdot r_2$ is stochastic.
\end{proposition}
\begin{proof}
$\displaystyle\sum_w\sum_{w=uv}r_1(u)r_2(v)
=\sum_u r_1(u)\sum_v r_2(v)=\SOne\cdot\SOne=\SOne.$
\end{proof}

The semantics of the star comes from a self-loop on a state. Call the
weight of the self-loop $t$ (\emph{continue}) and the weight of
leaving the state $e$ (\emph{stop}). Write $t^*:=\sum_{k\ge0}t^{k}$,
which is defined because $S$ is complete.

\begin{definition}[Closure equation]
\label{def:closure-eq}
The \emph{closure equation} relating $t$ and $e$ is $t^*e=\SOne$.
\end{definition}

Observe that whenever the closure equation $t^*e=\SOne$ has a solution, $e+t=\SOne$. This can be proved by using the identity $t^*=\SOne+t\,t^*$.
Multiplying by $e$ on the left gives $t^*e=e+t(t^*e)=e+t\SOne=e+t$, which proves the identity.

\begin{theorem}[Star]
\label{thm:star-general}
If $r$ is stochastic and guarded (\Cref{def:guarded}), then for every
$t\in S$ and $e\in S$ solving the closure equation $t^*e=\SOne$, the
expression $r^*_t$ is stochastic.
\end{theorem}
\begin{proof}
Summing over all $w\in\Sigma^*$ at once: a pair
$\big(k,(w_1,\ldots,w_k)\big)\in\{1,2,\ldots\}\times(\Sigma^+)^k$
determines $w=w_1\cdots w_k$ uniquely, so
\[
\sum_w r^*_t(w)=\sum_{k\ge1}\sum_{(w_1,\ldots,w_k)\in(\Sigma^+)^k}
r(w_1)\,t\,r(w_2)\,t\cdots t\,r(w_k)\,e.
\]
Fix $k$ and sum out $w_1,\ldots,w_k$ one at a time, left to right.
Summing over $w_1$ uses right-distributivity to pull
$\sum_{w_1}r(w_1)=\SOne$ out of the leftmost factor, leaving
$t\,r(w_2)\,t\cdots t\,r(w_k)\,e$; summing over $w_2$ pulls
$\sum_{w_2}r(w_2)=\SOne$ out from between the two copies of $t$ using
distributivity on both sides of that factor, leaving
$t\,t\,r(w_3)\cdots t\,r(w_k)\,e=t^2r(w_3)\cdots$; repeating for
$w_3,\ldots,w_k$ leaves $t^{k-1}e$. No copy of $t$ ever moves past an $r(w_i)$: each $r(w_i)$ sums to
$\SOne$ and disappears from its own position, using only left and right
distributivity, which hold in any semiring. Hence
\[
\sum_w r^*_t(w) = \sum_{k\ge1} t^{k-1}e = t^*e = \SOne
\]
by the closure equation.
\end{proof}
\newpage
\begin{remark}
\label{rem:hypotheses}
The axioms that \Cref{thm:star-general} need are:
\begin{itemize}
\item \emph{Guardedness.} Guardedness is required only for the
expression $r$ under a star $r^*_t$. The proof factors the sum over
decompositions into $k$ copies of $\sum_{w\in\Sigma^+}r(w)$, and this
sum equals $\SOne$ exactly when $r$ assigns no weight to the empty
word.
\item \emph{Distributivity only.} Choice, concatenation, and star rest
on left and right distributivity alone, and they hold verbatim over a
noncommutative $S$ (for example, matrices). No commutativity or
centrality assumption on $t$ is needed, because the proof of
\Cref{thm:star-general} only sums the terms $r(w_i)$ to $\SOne$ and
never moves a copy of $t$.
\item \emph{Positivity (or idempotence).} All three proofs reorder a
sum over an infinite index set ($\Sigma^*$ or $(\Sigma^+)^k$) without
changing its value. This is valid when $+$ is idempotent
(\Cref{lem:idempotent-positive}), since the rearrangement is then
associativity and commutativity of a semilattice join. For
non-idempotent $S$ such as $\mathbb{R}_{\geq0}$, it is valid when all
terms are nonnegative and the series is absolutely convergent. It fails
in general for signed $S$. Over $\mathbb{R}$, the machine with
$\lambda(q)=1$, $M_a(q,q)=2$, and $\mu(q)=-1$ is locally stochastic at
$q$, yet its weights $a\mapsto-2$, $aa\mapsto-4,\ldots$ sum to
$-\infty$. Without an absolute-summability hypothesis, local
stochasticity gives no control over the global sum.
\end{itemize}
\end{remark}

\Cref{thm:star-general} reduces the star construction to one scalar
equation, $t^*e=\SOne$, in two unknowns. The freedom this equation
leaves depends on whether $+$ is idempotent, as the three instances
below show.

\begin{enumerate}
\item \emph{$\mathbb{R}_{\geq0}$.} $t^*=\sum_{k\ge0}t^k=\frac{1}{1-t}$
for $t\in[0,1)$, so $t^*e=1\iff e=1-t$. The continuation probability $t$ is free, and the stop probability
is determined by $e=1-t$.
\item \emph{Tropical.} $t^*=\min_{k\ge0}kt=0$ for every $t\ge0$, since
the sequence $kt$ is nondecreasing in $k$, so the minimum is its
$k=0$ term. The equation $t^*e=\SOne$ becomes $0+e=0$, i.e.\ $e=0$
unconditionally, so $t$ is the only free parameter.
\item \emph{Viterbi.} For $t\in[0,1]$, the sequence
$t^0=1,t^1=t,t^2,\ldots$ is non-increasing, so $t^*=\max_{k\ge0}t^k=1$,
attained at $k=0$. The closure equation $t^*e=\SOne$ becomes
$1\cdot e=1$, forcing $e=1$ regardless of $t$.
\end{enumerate}

\section{The Conjugation Theorem}
\label{sec:conjugation}

We answer Q2 with one construction. We state it here for a general
semiring and instantiate it over $\Rp$, the tropical semiring, and the
Viterbi semiring in the following sections. Fix a
complete commutative semiring $S$ and a weighted automaton
$A=(\Sigma,Q,\lambda,M,\mu)$ with joint transition matrix
$M=\sum_{a\in\Sigma}M_a$. Partitioning the mass by word length gives
\[
\|f_A\| \;=\; \sum_{w\in\Sigma^*} f_A(w)
\;=\; \sum_{n\ge0}\lambda^\top M^n\mu
\;=\; \lambda^\top M^*\mu,
\qquad M^* := \sum_{n\ge0} M^n,
\]
where $M^*$ is defined entrywise by completeness.

\begin{definition}[Completion vector]
\label{def:completion}
The \emph{completion vector} of $A$ is $d := M^*\mu$. Its entry $d(q)$ is mass of $A_q$, the automaton $A$ with its initial vector replaced by
the unit vector at $q$.
\end{definition}

Observe that $M^* = I + MM^*$, and multiplying by $\mu$ on the
right gives us the closure identity $d = \mu + Md$.

\begin{theorem}[Conjugation]
\label{thm:conjugation}
Let $A$ be trim, and suppose that $Z := \|f_A\| = \lambda^\top d$ and
every entry $d(q)$ of the completion vector have multiplicative inverses
in $S$. Then the automaton $\mathcal N := (\Sigma,Q,\lambda',M',\mu')$
with
\[
\lambda'(q) = Z^{-1}\,\lambda(q)\,d(q), \qquad
M'_a(q,q') = d(q)^{-1}\,M_a(q,q')\,d(q'), \qquad
\mu'(q) = d(q)^{-1}\,\mu(q)
\]
is locally stochastic (\Cref{def:local-stoch-general}), has the same
state set, and satisfies $f_A(w) = Z \cdot f_{\mathcal N}(w)$ for every
$w\in\Sigma^*$.
\end{theorem}
\begin{proof}
Local stochasticity at $q$: factoring $d(q)^{-1}$ out of the finite sum,
\[
\sum_{a\in\Sigma}\sum_{q'\in Q} M'_a(q,q') + \mu'(q)
= d(q)^{-1}\Bigl(\,\sum_{a,q'} M_a(q,q')\,d(q') + \mu(q)\Bigr)
= d(q)^{-1}\,d(q) = \SOne
\]
by the closure identity, and
$\sum_q\lambda'(q) = Z^{-1}\lambda^\top d = \SOne$. For the semantics,
along a single run $q_0\xrightarrow{w_1}q_1\cdots\xrightarrow{w_n}q_n$
the factors $d(q_i)\,d(q_i)^{-1}$ telescope, so
\[
\lambda'(q_0)\prod_{i=1}^n M'_{w_i}(q_{i-1},q_i)\;\mu'(q_n)
= Z^{-1}\,\lambda(q_0)\prod_{i=1}^n M_{w_i}(q_{i-1},q_i)\;\mu(q_n);
\]
summing over the finitely many runs labelled $w$ gives
$f_{\mathcal N}(w)=Z^{-1}f_A(w)$.
\end{proof}

The theorem leaves each semiring two obligations. First, the entries
of $d$ must be invertible in the semiring. This is a
spectral condition: $\rho(M)<1$ over $\Rp$ (\Cref{prop:finite-mass}),
nonnegative cycle costs over the tropical semiring, and automatic over
the Viterbi semiring. When it fails, a growth-rate normalisation
restores it. Second, the inverses must be given explicitly: division over
$\Rp$ and the Viterbi semiring, and subtraction over the tropical
semiring. Neither follows from the axioms used in
\Cref{sec:rung0,sec:rung2}, so each semiring has to provide them.

\section{Conversions in the $\Rp$ semiring}
\label{sec:realisability}

This section treats Q2 over $\Rp$. We call a locally stochastic
automaton (\Cref{def:local-stoch-general}) over $\Rp$ a
\emph{probabilistic automaton}, the generative model of
\Cref{sec:relwork}. An automaton of finite mass $Z$ reduces to the
mass-one case by dividing its initial vector by $Z$, which is computable
by \Cref{prop:finite-mass}.

The construction is \Cref{thm:conjugation} applied with $d=M^*\mu$ to
the whole automaton. If $\rho(M)<1$ and every $d(q)$ is invertible, it
gives a locally stochastic rescaling in one step, whatever the number of
SCCs. The rest of this section refines how $d$ is computed. The block
structure confines matrix inversion to the SCC blocks $M_{ii}$, which
lowers the cost from the $n^3$ of one global inversion to
$\sum_i n_i^3\le n^3$ for SCC sizes $n_i$, and the passage between SCCs
needs only local arithmetic. The same two-level scheme works over the tropical and Viterbi
semirings.

\subsection{Finite Mass Characterisation}
Consider the semiring $(\mathbb{R}_{\geq0}, +, \times, 0, 1)$, and let
$M = \sum_{a\in \Sigma} M_a$ be the joint transition matrix. The total
mass is then

\[
\|f_A\| = \sum_{w \in \Sigma^*} f_A(w) = \sum_{n=0}^{\infty} \sum_{w \in \Sigma^n} \lambda^{\top} M_{w_1} \cdots M_{w_n} \mu = \sum_{n=0}^{\infty} \lambda^{\top} M^n \mu .
\]

\begin{proposition}[Spectral characterisation of finite mass]
\label{prop:finite-mass}
Let $A$ be a trim weighted automaton. Then
$\|f_A\|<\infty$ iff $\rho(M)<1$.
\end{proposition}
\begin{proof}
Since $M \geq 0$ (entrywise), we have $M^n \geq 0$ for all $n$, and thus:
\[
\|f_{A}\| = \lambda^{\top} \left(\sum_{n=0}^{\infty} M^n\right) \mu = \lambda^{\top} \left(I - M\right)^{-1} \mu \qquad\text{if  } \rho(M) < 1.
\]

This follows from the Neumann series theorem (Appendix~\ref{app:matrix-properties}).
Since $\lambda, \mu \geq 0$ and $(I - M)^{-1} \geq 0$ when $\rho(M) < 1$,
the total mass is finite. Conversely, if $\rho(M) \geq 1$, some entry
$(i,j)$ of $\sum_{n} M^n$ diverges (Neumann series again). By trimness,
there is a path of positive weight $c_1$ from some state $p$ with
$\lambda(p)>0$ to $i$, and a path of positive weight $c_2$ from $j$ to
some state $q$ with $\mu(q)>0$. Since all terms are nonnegative,
$\|f_A\| \geq \lambda(p)\, c_1 \bigl(\textstyle\sum_n M^n\bigr)_{ij}\, c_2\, \mu(q) = \infty$.
\end{proof}

\begin{remark}[Generalisability of the spectral characterisation]
    We do not claim a semiring-general version of this proposition.
    Over $\Rp$, the equivalence rests on the Neumann series and the
    invertibility of $I-M$, which have no counterpart in an arbitrary
    semiring. The other two instances in this paper share the pattern
    that finiteness of mass is governed by a spectral quantity of the
    joint transition matrix. Over the tropical semiring, the role of
    $\rho(M)$ is played by the cycle mean (\Cref{def:cycle-mean}), and
    over the Viterbi semiring the mass is always finite
    (\Cref{lem:viterbi-mass}), so the corresponding condition holds
    automatically. Even stating a common generalisation requires a total
    order and $n$-th roots of arbitrary elements, which motivates the
    following definition.
\end{remark}

\begin{definition}[Radicable semiring]
\label{def:radicable}
An ordered semiring $(S, +, \times, \SZero, \SOne, \leq)$ is \emph{radicable} if for every $a \in S$ and every $n \in \mathbb{N}_{\geq 1}$ there exists a unique element $a^{1/n} \in S$, the \emph{$n$-th root} of $a$, satisfying
\[
\underbrace{a^{1/n} \times a^{1/n} \times \cdots \times a^{1/n}}_{n \text{ times}} = a,
\]
and the map $(\cdot)^{1/n} : S \to S$ is monotone with respect to $\leq$, i.e. $a \leq b \implies a^{1/n} \leq b^{1/n}$.
\end{definition}
In the rest of this section we work with automata of finite mass over
$\Rp$. The remarks after each step identify the axioms that step needs.

\subsection{Strongly Connected Components and Block Structure}

The normalisation has two phases: the first normalises each strongly
connected component (SCC) of the automaton, and the second normalises
the acyclic structure connecting them.

Let $G = (Q, E)$ be the directed graph with $(q, q') \in E$ if and only
if $\sum_{a \in \Sigma} M_a(q, q') > 0$, and let $Q_1, Q_2, \ldots, Q_k$
be its strongly connected components. After permuting the states, $M$
has block upper-triangular form:
\[
M = \begin{pmatrix}
M_{11} & M_{12} & \cdots & M_{1k} \\
0 & M_{22} & \cdots & M_{2k} \\
\vdots & \vdots & \ddots & \vdots \\
0 & 0 & \cdots & M_{kk}
\end{pmatrix}
\]
where $M_{ii}$ is the restriction of $M$ to the $i$-th SCC $Q_i$, and
$M_{ij}$ for $i < j$ holds the transitions from $Q_i$ to $Q_j$. The
permutation is a topological sort of the SCCs.

\begin{figure}[t]
    \centering
    \begin{tikzpicture}[->,auto,accepting/.style=accepting by arrow, thick, node distance=2.2cm]
    \node[state,initial,initial text={},initial distance=8mm] (q0) {$q_0$};
    \node[state,accepting,accepting distance=8mm] (q4) at (5,1.5) {$q_4$};
    \node at (-0.8,0.3) {$1$};
    \node at (5.8,1.8) {$1$};
    \node at (5,0.3) {$2a$};
    \node at (3.8,-0.6) {$\frac{3}{5}a$};
    \node at (2.9,-1.2) {$\frac{2}{5}b$};
    \node at (0.9,-1) {$\frac{2}{5}a$};
    \node[state] (q3) at (7,0) {$q_3$};
    \node[state] (q1) at (2,-1.5) {$q_1$};
    \node[state] (q2) at (5,-1.5) {$q_2$};
    \node[state] (q5) at (2,1.5) {$q_5$};
    \path[-stealth,thick]
    (q0) edge node {} (q1)
    (q1) edge node {} (q2)
    (q2) edge[loop below] node[below] {$\frac{3}{5}a$} ()
    (q0) edge[bend left=10] node {} (q2)
    (q0) edge node {} (q3)
    (q2) edge node[below] {$2a$} (q3)
    (q3) edge node[above] {$3b$} (q4)
    (q0) edge node[above] {$1a$} (q5)
    (q5) edge node[above] {$2a$} (q4)
    (q5) edge[loop above] node {$\frac{1}{3}b$} ()
    (q2) edge[bend left=20] node {$\frac{2}{5}a$} (q1);
    \end{tikzpicture}
    \caption{A weighted automaton $A$ with alphabet $\Sigma = \{a, b\}$
    and total mass $\|f_{A}\| = 28$. After normalisation by $1/28$ the
    automaton is \emph{globally} stochastic, but it is not \emph{locally}
    stochastic: state $q_0$ has outgoing weight
    $\frac{2}{5} + \frac{3}{5}+ 2 + 1 > 1$. The SRE for this machine is Equation \ref{eq:sre-big}}
    \label{fig:running-example}
\end{figure}

In the automaton of \Cref{fig:running-example}, the strongly connected
components are $\{q_0\}$, $\{q_1,q_2\}$, $\{q_3\}$, $\{q_4\}$, and
$\{q_5\}$.

The blocks are the SCCs of the graph of the automaton. The decomposition holds over every semiring, because it depends only on which entries of $M$ are nonzero. After a topological ordering, $M$ is block upper triangular with
diagonal blocks $M_{ii}$, so the spectrum of $M$ is the union of the
spectra of the blocks.

Hence, For each SCC of the automaton, $\rho(M_{ii}) \leq \rho(M)$. Moreover, $\rho(M) < 1$ because of the finite mass condition.

\subsection{Normalising Strongly Connected Components}
\label{sec:scc-normalisation}

Write $Q_i$ for the $i$-th SCC, with block $M_i = M_{ii}$, and let
$d = M^*\mu$ be the completion vector of \Cref{def:completion}. For
$q \in Q_i$, split $d(q)$ into the part that stays in $Q_i$ and the part
that leaves it. The \emph{exit vector} of $Q_i$ is
\[
\varepsilon_i(q) \;=\; \mu(q) \;+\; \sum_{j>i}\ \sum_{a\in\Sigma}\ \sum_{q'\in Q_j} M_a(q,q')\,d(q'),
\qquad q\in Q_i,
\]
the final weight of $q$ plus the weighted completion values of the
transitions leaving $Q_i$. Let $d_i$ denote the restriction of $d$ to
$Q_i$. The identity $Md + \mu = d$, on the
block $Q_i$, gives $d_i = M_i d_i + \varepsilon_i$, hence
\[
d_i \;=\; (I - M_i)^{-1}\varepsilon_i,
\]
where $(I-M_i)^{-1}$ exists because $\rho(M_i)<1$. The only matrix inversions are
therefore the inversions of the blocks $I-M_i$. The exit vectors depend
only on $d_j$ for $j>i$, so $d$ is obtained by processing the SCCs from
$Q_k$ down to $Q_1$.

\begin{lemma}[Well-definedness of $d$]
\label{lem:d-well-defined}
Let $A$ be a trim automaton of finite mass over $\Rp$. Then $d_i =
(I-M_i)^{-1}\varepsilon_i$ is well defined for every $i$, $d = M^*\mu$,
and $d(q)>0$ for every state $q$.
\end{lemma}

\begin{proof}
Since $\rho(M_i)<1$, we have $(I-M_i)^{-1}=\sum_{n\ge0}M_i^n\ge0$ by the
Neumann series, and $\varepsilon_i\ge0$ by induction from $i=k$ down to
$i=1$, so $d_i\ge0$. The vector assembled from the $d_i$ satisfies
$d=Md+\mu$ on every block, and $I-M$ is invertible since $\rho(M)<1$. Positivity follows from trimness: every state
reaches a state with $\mu>0$, so $d(q)>0$.
\end{proof}

Let $D=\operatorname{diag}(d)$ and define
\[
\lambda' = \lambda D\,Z^{-1}, \qquad M'_a = D^{-1} M_a D, \qquad \mu' = D^{-1}\mu,
\]
as in \Cref{thm:conjugation}. This normalisation preserves semantics:
the factors $D$ and $D^{-1}$ telescope along every run.

\begin{remark}[Axioms needed]
Over the reals, $d_i$ is computed by $d_i = (I - M_i)^{-1}\varepsilon_i$.
Over an arbitrary positive semiring it can instead be defined as
$d_i = (\sum_{n=0}^{\infty} M_i^n)\,\varepsilon_i$, which converges
because the spectral radius is less than $\SOne$. The tropical case
gives an equivalent definition. Computing $d$ therefore needs no additional
axioms.

The diagonal scaling does need $D$ to be invertible, so every entry
$d(q)$ must have a multiplicative inverse. The semiring must therefore
be a \emph{semifield}, or at least every nonzero entry of $d$ must be
invertible. In practice one may discard the states with $d(q)=\SZero$.
\end{remark}
By definition of $d_i$, $M_i d_i + \varepsilon_i = d_i$. Therefore, for every state $q \in Q_i$,

\begin{equation}
\label{eq:scc-local-stochastic}
\sum_{a \in \Sigma} \sum_{q' \in Q_i} M'_a(q, q') + \varepsilon'_i(q) = 1,
\qquad\text{where } \varepsilon'_i(q) = \varepsilon_i(q)/d(q).
\end{equation}

\subsection{Normalising the Acyclic Structure}
\label{sec:acyclic}

Collapse each SCC to a single \emph{super-state}. The result is acyclic,
since SCCs contain all cycles, and a super-state has no self-loop: every
transition between distinct SCCs leaves the block it starts in. The
self-loops and cycles inside each SCC are already accounted for by the
factor $(I-M_i)^{-1}$ of \Cref{sec:scc-normalisation}, which plays the
role of the geometric factor $1/(1-S(q))$ for a single state. In the
terms of \Cref{rem:acyclic-general}, $S(q)=0$ on the quotient, and the
closure equation $S(q)^*e(q)=\SOne$ gives $e(q)=\SOne$.

\Cref{alg:acyclic-normalisation} handles the transitions between SCCs.
It processes the super-states in reverse topological order, computes the
exit vector of each block from the completion values of the blocks
already processed, and rescales the inter-SCC transitions. It needs only
local arithmetic on these transitions. The algorithm follows the
classical conversion of finite automata to regular
expressions~\cite{HMU06}.

Define a partial order on SCCs by $Q_i \prec Q_j$ if some path leads
from a state in $Q_i$ to a state in $Q_j$ without staying within $Q_i$.
Since the quotient graph is acyclic, we can order $Q_1, \ldots, Q_k$ so
that $Q_i \prec Q_j$ implies $i < j$.

\begin{algorithm}[t]
\caption{Acyclic normalisation of the SCC quotient}
\label{alg:acyclic-normalisation}
\begin{algorithmic}[1]
\Require Transition matrices $\{M_a\}_{a \in \Sigma}$, vectors $\lambda,\mu$, SCCs $Q_1,\ldots,Q_k$ in topological order, and the blocks $(I-M_{ii})^{-1}$
\Ensure Normalised $\lambda'$, $\{M'_a\}_{a \in \Sigma}$, $\mu'$, and the mass $Z$

\For{$i = k$ \textbf{down to} $1$}
    \Comment{Reverse topological order}
    \For{each state $q \in Q_i$}
        \State $\varepsilon_i(q) \gets \mu(q) + \sum_{a \in \Sigma} \sum_{j>i} \sum_{q' \in Q_j} M_a(q,q')\, d(q')$
        \Comment{Exit weight; no self-loop term, $S(q)=0$}
    \EndFor
    \State $d_i \gets (I - M_{ii})^{-1}\varepsilon_i$
    \Comment{Absorbs all loops inside $Q_i$}
\EndFor
\State $Z \gets \sum_q \lambda(q)\,d(q)$
\For{each $a\in\Sigma$ and states $q,q'$}
    \State $M'_a(q,q') \gets M_a(q,q')\,d(q')\,/\,d(q)$
\EndFor
\State $\mu'(q) \gets \mu(q)/d(q)$ and $\lambda'(q) \gets \lambda(q)\,d(q)/Z$ for every $q$
\end{algorithmic}
\end{algorithm}

\begin{figure}
    \centering
\begin{tikzpicture}[->,auto,accepting/.style=accepting by arrow, thick]
    \node [state, initial, initial text = {}, initial distance = 8mm](q02) at (0,-2) {$q_0$};
    \node(q32)[state] at (2.5,-2) {$q_1$};
    \node(q42)[state] at (5,-2) {$q_2$};
    \node(q52)[state] at (7.5,-2) {$q_3$};   
    \node(q62)[state, accepting, accepting text = {}, accepting distance = 8mm] at (10,-2) {$q_4$};
    \node at (-0.8, -1.75) {$\frac 4 {28}$};
    \node at (10.8, -1.75) {$1$};
    \path[-stealth, thick]
    (q02) edge node [above] {$1a$} (q32)
    (q42) edge node [above] {$\frac {6} {25}a$} (q52)
    (q52) edge node [above] {$1b$} (q62)
    (q32) edge [bend left = 15] node [above] {$1b$} (q42)
    (q42) edge [bend left = 15] node [below] {$\frac 4 {25} a$} (q32)
    (q42) edge [loop above] node [above] {$\frac 35a$}();
    \end{tikzpicture}
    \caption{The weighted automaton from Figure \ref{fig:running-example} can be written down as acyclic components according to its topological sorting. The machine in the figure shows the application of Algorithm \ref{alg:acyclic-normalisation} on the automata corresponding to the partial order $ \{q_0\} \prec \{q_1,q_2\} \prec \{q_3\} \prec \{q_4\}$. The SRE for this machine is Equation \ref{eq:sre-small}} 
    \label{fig:after-normalisation}
\end{figure}

\begin{remark}[Acyclic normalisation over a general semiring]
\label{rem:acyclic-general}
\Cref{alg:acyclic-normalisation} works over any semiring in which the
blocks $I-M_{ii}$ can be replaced by the closures $M_{ii}^*$ and every
$d(q)$ has a multiplicative inverse. Where a single state carries a
self-loop of mass $S(q)=\sum_a M_a(q,q)$, solve $S(q)^*e(q)=\SOne$ for
$e(q)$ in place of $1-S(q)$. On the quotient this term is $\SOne$.
\end{remark}

\begin{lemma}[Acyclic normalisation preserves semantics]
\label{lem:acyclic-preserves-semantics}
\Cref{alg:acyclic-normalisation} computes $d=M^*\mu$ and its output
satisfies $f_{A'} = Z^{-1} f_A$.
\end{lemma}

\begin{proof}
We show first that the loop on lines~1--6 computes $d=M^*\mu$, by
induction on $k-i$. For $i=k$ there are no later blocks, so
$\varepsilon_k=\mu_k$ and $d_k=(I-M_{kk})^{-1}\mu_k$. If $d_j$ is
correct for all $j>i$, then $\varepsilon_i$ is the exit vector of
\Cref{sec:scc-normalisation}, so $d_i=(I-M_{ii})^{-1}\varepsilon_i$ is
the $i$-th block of $M^*\mu$ by \Cref{lem:d-well-defined}. For the
semantics, the factors $d(q_i)$ telescope along every run, exactly as in
the proof of \Cref{thm:conjugation}, so each run of $A$ on $w$ has
weight $Z$ times the weight of the corresponding run of $A'$. Summing
over runs gives $f_A(w)=Z\,f_{A'}(w)$.
\end{proof}

\begin{theorem}[Local stochasticity after acyclic normalisation]
\label{thm:acyclic-local-stochastic}
After \Cref{alg:acyclic-normalisation}, for every state $q$,
\[
\sum_{a \in \Sigma} \sum_{q' \in Q} M'_a(q, q') + \mu'(q) = 1,
\qquad \sum_{q}\lambda'(q)=1.
\]
\end{theorem}

\begin{proof}
Split the sum over $q'$ into $q'\in Q_i$ and $q'\in Q_j$ for $j>i$,
where $q\in Q_i$. Transitions into earlier blocks have weight zero. By
the definition of $\varepsilon_i$ and \Cref{eq:scc-local-stochastic},
\[
\sum_{a}\sum_{q'} M'_a(q,q') + \mu'(q)
= \frac{1}{d(q)}\Bigl(\sum_a\sum_{q'\in Q_i}M_a(q,q')d(q')
+ \varepsilon_i(q)\Bigr)
= \frac{d(q)}{d(q)} = 1.
\]
Also $\sum_q\lambda'(q)=Z^{-1}\sum_q\lambda(q)d(q)=1$.
\end{proof}

\subsection{Finite Mass Normalisation}

Combining the two steps
(\Cref{eq:scc-local-stochastic,thm:acyclic-local-stochastic}) gives the
following theorem.

\begin{theorem}[Finite-mass normalisation]
\label{thm:finite-mass-normal-form}
Every trim weighted automaton $A$ over $\mathbb{R}_{\geq 0}$ with
finite mass $Z$ has an equivalent probabilistic automaton $A'$ with
$f_A = Z\, f_{A'}$. In particular, if $Z=1$ then $f_A = f_{A'}$.
$A'$ can be constructed in $O(n^3 + n^2 |\Sigma|)$ time, and has the
same number of states as $A$.
\end{theorem}

\begin{proof}
Compute $d$ by SCC normalisation and
\Cref{alg:acyclic-normalisation}. Semantics are preserved by
\Cref{lem:acyclic-preserves-semantics}, and local stochasticity holds by
\Cref{thm:acyclic-local-stochastic}.

SCC decomposition takes $O(n + n|\Sigma|)$ time. For each SCC of size
$n_i$, computing $(I - M_{ii})^{-1}$ takes $O(n_i^3)$ time. Since
$n_i\le n$, $n_i^3 \le n_i n^2$, so $\sum_i n_i^3 \le n^2\sum_i n_i =
n^3$, giving $O(n^3)$ for all block inversions; the bound is smaller
when the SCC sizes are unbalanced. Computing the exit vectors and
rescaling touch each of the $O(n^2|\Sigma|)$ transition entries a
constant number of times. Total: $O(n^3+n^2|\Sigma|)$.
\end{proof}

Without the final vector, $A'$ is a labelled Markov chain, so
techniques from the literature on probabilistic automata and Markov
chains apply. Recent work on labelled Markov chains and weighted
automata uses this construction implicitly~\cite{chistikov2020bigo,
chistikov2022bigo}. We state the construction and its algorithms
explicitly. The normalisation makes the toolkit of probabilistic
automata theory applicable to any weighted automaton with finite mass.

\subsection{Conversion to stochastic regular expression}
\begin{definition}[Stochastic regular expressions for the $\mathbb{R}_{\geq0}$ semiring]
Instantiating \Cref{def:sre-syntax} over $(\mathbb{R}_{\geq 0}, +, \times, 0, 1)$ with the closure solution $e = 1-t$:
\begin{itemize}
  \item $\delta_w(w') = 1$ if $w = w'$, else $0$
  \item $(a_1 r_1 + a_2 r_2)(w) = a_1 r_1(w) + a_2 r_2(w)$ \quad such that $a_1 + a_2 = 1$
  \item $(r_1 \cdot r_2)(w) = \sum_{w = uv} r_1(u) \, r_2(v)$
  \item $r^*_{t}(w) = \sum_{k \geq 1} \sum_{w = w_1 \cdots w_k} t^{k-1} \, (1-t) \prod_{i=1}^k r(w_i)$, \quad $t \in [0,1)$, $w_i \in \Sigma^+$
\end{itemize}
All of these maintain the invariant that $\sum_{w \in \Sigma^*} r(w) = 1$.
\end{definition}
\begin{theorem}[State elimination for probabilistic automata]
\label{thm:state-elimination}
A probabilistic automaton $A$ can be converted to an SRE $r$ such that $\llbracket r \rrbracket = \llbracket A \rrbracket$. The size of $r$ may be exponential in $|A|$.
\end{theorem}

\begin{proof}
We adapt the classical state elimination algorithm~\cite{kleene1956}.
Local stochasticity ensures that every self-loop has weight strictly
less than $1$, so each geometric series below converges.

Given a normalised weighted automaton
$\mathcal{W} = (\Sigma,Q,\lambda,M,\mu)$, augment it with a new start
state $q_s$ and final state $q_f$, and work with edge labels that are
nonnegative scalar combinations of SREs: initially
$R_{ij}=\sum_{\sigma} M_\sigma(q_i,q_j)\,\delta_\sigma$ for
$q_i,q_j\in Q$, $R_{sj}=\lambda(q_j)\,\delta_\varepsilon$, and
$R_{if}=\mu(q_i)\,\delta_\varepsilon$. To eliminate an intermediate state
$q_k\in Q$, let $s_k=\sum_\sigma M_\sigma(q_k,q_k)$ be its total
self-loop weight and, if $s_k>0$, let
\[
L_k=\sum_{\sigma} \frac{M_\sigma(q_k,q_k)}{s_k}\,\delta_\sigma
\]
be its loop expression, an SRE of mass $1$, guarded since each
$\delta_\sigma$ has $\sigma\in\Sigma$. Replace every label $R_{ij}$ by
\[
R_{ij}\;+\;R_{ik}\cdot\Bigl(\delta_\varepsilon
+ \tfrac{s_k}{1-s_k}\,(L_k)^{*}_{s_k}\Bigr)\cdot R_{kj},
\]
omitting the second summand of the bracket when $s_k=0$, and delete
$q_k$. The two summands account for leaving $q_k$ without, respectively
after at least one, traversal of its self-loops: on $m\ge1$ traversals
the discounted star $(L_k)^{*}_{s_k}$ carries the factor
$s_k^{m-1}(1-s_k)$, and the scalar $s_k/(1-s_k)$ turns this into the
required $s_k^{m}$. Because $s_k<1$, all quantities are well defined.
After all eliminations, the remaining label on $q_s \to q_f$ is an
expression $r$ with $\llbracket r \rrbracket = \llbracket \mathcal{W}
\rrbracket$; its total mass is $1$ because elimination preserves the
semantics and $\mathcal{W}$ is stochastic, so the scalar combinations
assemble into an SRE.
\end{proof}

\subsection{Main Equivalence Theorem}

\Cref{thm:state-elimination} turns a probabilistic automaton into an SRE,
and the Thompson construction of \Cref{app:sre} turns an SRE into a
probabilistic automaton with linearly many states. Together they show that
the two formalisms describe the same class of word functions. Combined
with the normalisation of \Cref{thm:finite-mass-normal-form}, we get:

\begin{theorem}[Kleene--Schützenberger theorem for stochastic languages]
\label{thm:kleene-main}
A function $f : \Sigma^* \to [0,1]$ is computable by a probabilistic automaton if and only if it is expressible as a stochastic regular expression.
\end{theorem}

The machine in Figure \ref{fig:after-normalisation} has the corresponding SRE:
\begin{equation}
\label{eq:sre-small}
\frac 4 {28} \Big(ab\Big(\frac 4 {19}ab + \frac{15}{19} a\Big)^*_{\frac {19}{25}}\Big)ab    
\end{equation}

The stochastic regular expression for the entire machine in Figure \ref{fig:running-example} reduces to:

\begin{equation}
\label{eq:sre-big}
\frac {19} {28} \Big(\Big(\frac{4}{19}ab + \frac{15}{19}a\Big)\Big(\frac 4 {19}ab + \frac{15}{19} a\Big)^*_{\frac {19}{25}}\Big)ab \,\, + \frac 6 {28} ab \,\, + \frac 3 {28} a\Big(b\Big)^*_{\frac 1 3}a
\end{equation}

\section{Normalisation Beyond Finite Mass}
\label{sec:beyond-finite-mass}

When $\rho(M)\geq1$, the mass is infinite and
\Cref{thm:finite-mass-normal-form} does not apply. Rescaling the
transition weights by the growth rate restores finite mass.

\begin{definition}[Spectral normalisation]
\label{def:spectral-normalisation}
Let $A$ be a weighted automaton over $\mathbb{R}_{\geq 0}$ with spectral radius $\rho(M)$. For any $\epsilon > 0$, define the \emph{discount factor} $\zeta = (1 + \epsilon) \rho(M)$. The \emph{spectrally normalised automaton} $A_{\zeta}$ is obtained by scaling each transition matrix:
\[
M'_a = \frac{1}{\zeta} M_a \quad \text{for all } a \in \Sigma
\]
while keeping $\lambda$ and $\mu$ unchanged.
\end{definition}

By construction, $A_\zeta$ satisfies $f_{A_{\zeta}}(w) = \zeta^{-|w|}
f_{A}(w)$ for all $w \in \Sigma^*$, and its spectral radius is
$\rho(M') = \rho(M)/\zeta = 1/(1+\epsilon) < 1$. Its total mass
$Z = \|f_{A_{\zeta}}\|$ is therefore finite, and it is positive when
$A$ is trim. Dividing the initial vector of $A_{\zeta}$ by $Z$ gives an
automaton of total mass $1$, to which
\Cref{thm:finite-mass-normal-form} applies. This yields a probabilistic
automaton $\mathcal{N}$ with $f_{\mathcal{N}} = \frac{1}{Z}
f_{A_{\zeta}}$, and by \Cref{thm:kleene-main}, $f_{\mathcal{N}}$ is
expressible as an SRE $r$. Hence:

\begin{theorem}[Spectral decomposition]
\label{thm:spectral-decomposition}
Every trim weighted automaton $A$ over $\mathbb{R}_{\geq 0}$ with $\rho(M)>0$ admits a \emph{tripartite decomposition}:
\[
f_{A}(w) = \zeta^{|w|} \cdot Z \cdot \llbracket r \rrbracket(w)
\]
where:
\begin{itemize}
\item $\zeta = (1+\epsilon) \rho(M)$ for arbitrarily small $\epsilon > 0$ (exponential growth rate)
\item $Z = \|f_{A_{\zeta}}\| \in (0, \infty)$ (residual total mass)
\item $r$ is an SRE with $\sum_w \llbracket r \rrbracket(w) = 1$ (stochastic part)
\end{itemize}
\end{theorem}

When $\rho(M) = 0$, the mass is already finite, and the decomposition
holds with $\zeta = 1$ directly by
\Cref{thm:finite-mass-normal-form,thm:kleene-main}. The decomposition
covers every trim weighted automaton over $\mathbb{R}_{\geq 0}$.

\section{Conversions in the tropical semiring}
\label{sec:realisability-trop}

We treat the tropical semiring in the same way as $\Rp$.

\begin{definition}
\label{def:cycle-mean}
For the min-plus semiring matrix $M : Q \times Q \to \mathbb{R} \cup \{\infty\}$, the
\emph{tropical spectral radius}, or \emph{cycle mean}, is defined as
\[
\rho_{\text{trop}}(M) =
\min_{k \geq 1} \;
\min_{q_1, \ldots, q_k}
\frac{1}{k}
\left(
M(q_1, q_2) + \cdots + M(q_{k-1}, q_k) + M(q_k, q_1)
\right),
\]
where the inner minimum ranges over all cycles of length $k$.
\end{definition}

The cycle mean is the minimal average cost per transition along a
cycle, and it governs the asymptotic behaviour of long paths.
\begin{definition}[Tropical spectral normalisation]
\label{def:tropical-normalisation}
Let $A$ be a tropical automaton and let $M = \min_a M_a$ be its joint transition
matrix (in the tropical sense). Let $\gamma = \rho_{\text{trop}}(M)$ be its cycle
mean. The \emph{tropical spectral normalisation} of $A$ is the automaton $A_1$
defined by $
M'_a(q,q') = M_a(q,q') - \gamma$, for all $a \in \Sigma$ and $q,q' \in Q$.
\end{definition}

By construction, every cycle in the normalised automaton $A_1$ has
nonnegative total cost, and at least one cycle has total cost~$0$. A
tropical automaton $A$ is in \emph{tropically normalised form} if
$\rho_{\text{trop}}(M) = 0$, $\min_{w \in \Sigma^*} f_A(w) = 0$, and
$f_A(w_0) = 0$ for some $w_0 \in \Sigma^*$.

\begin{theorem}[Tropical decomposition]
\label{thm:tropical-decomposition}
Every trim tropical weighted automaton $A$ whose transition graph contains
a cycle admits a spectral decomposition of the form
\[
f_A(w) = |w| \cdot \gamma + c_0 + f_{\mathcal{N}}(w),
\]
where:
\begin{itemize}
\item $\gamma = \rho_{\text{trop}}(M)$ is the cycle mean (linear growth rate),
\item $c_0 = \min_{w \in \Sigma^*} \bigl(f_A(w) - \gamma\,|w|\bigr)$ is the
      minimal attainable cost once the growth term is removed; it is finite
      and attained,
\item $\mathcal{N}$ is a tropical automaton in tropically normalised form with
      $\min_w f_{\mathcal{N}}(w) = 0$.
\end{itemize}
\end{theorem}

\begin{proof}
Let $A_1$ be the tropical spectral normalisation of $A$
(\Cref{def:tropical-normalisation}). Subtracting $\gamma$ from every
transition subtracts $\gamma|w|$ from every run on $w$, so
$f_{A_1}(w) = f_A(w) - \gamma|w|$ and $c_0 = \min_w f_{A_1}(w)$. Every
cycle of $A_1$ has nonnegative cost, so deleting a cycle from a run cannot
increase its cost, and $\min_w f_{A_1}(w)$ is therefore attained by one of
the finitely many runs that repeat no state; it is finite because a trim
automaton has at least one run of finite cost. Let $\mathcal N$ be $A_1$
with initial vector $\lambda_{\mathcal N}(q) = \lambda(q) - c_0$. Then
$f_{\mathcal N} = f_{A_1} - c_0$, so $\min_w f_{\mathcal N}(w) = 0$ and
the minimum is attained; the transition matrices are those of $A_1$, so
$\rho_{\text{trop}} = 0$. Hence $\mathcal N$ is in tropically normalised form
and $f_A(w) = |w|\gamma + c_0 + f_{\mathcal N}(w)$.
\end{proof}

If the transition graph of $A$ is acyclic, $f_A$ is supported on words
of bounded length. The decomposition then holds without the growth term,
with $c_0 = \min_w f_A(w)$.

As over $\Rp$, the normalisation can be made local. The completion
vector is replaced by the \emph{least completion cost}
$d(q) = \min_{w\in\Sigma^*} \semantics{(A_1)_q}(w)$, the cheapest cost
of a run from $q$, where $(A_1)_q$ is $A_1$ started at $q$. On a trim
automaton, $d(q)$ is finite and is attained by a run that repeats no
state, by the same cycle-deletion argument. It satisfies the Bellman
equation
\[
d(q) \;=\; \min\Bigl(\min_{a\in\Sigma}\min_{q'\in Q}\,
M'_a(q,q') + d(q'),\ \mu(q)\Bigr),
\]
where $M'_a$ are the transition matrices of $A_1$
(\Cref{def:tropical-normalisation}).

\begin{theorem}[Tropical local normalisation]
\label{thm:tropical-local}
The automaton $\mathcal N$ of \Cref{thm:tropical-decomposition} can be
chosen locally stochastic in the sense of \Cref{def:local-stoch-general}
(the sums there are minima): taking
\[
\lambda_{\mathcal N}(q) = \lambda(q) + d(q) - c_0, \qquad
M^{\mathcal N}_a(q,q') = M'_a(q,q') + d(q') - d(q), \qquad
\mu_{\mathcal N}(q) = \mu(q) - d(q),
\]
we have $\min_q \lambda_{\mathcal N}(q) = 0$ and, for every state $q$,
$\min\bigl(\min_{a,q'} M^{\mathcal N}_a(q,q'),\ \mu_{\mathcal N}(q)\bigr)
= 0$, while $f_{\mathcal N} = f_{A_1} - c_0$ and all cycle costs are
unchanged.
\end{theorem}

\begin{proof}
Subtracting $d(q)$ from both sides of the Bellman equation gives the
displayed local condition. Splitting a run at its first state gives
$\min_q \bigl(\lambda(q) + d(q)\bigr) = \min_w f_{A_1}(w) = c_0$, hence
$\min_q \lambda_{\mathcal N}(q) = 0$. Along any run the $d$-terms
telescope, so the value of every run, and hence $f_{\mathcal N}$, changes
by exactly $-c_0$ relative to $A_1$. Cycle costs are unchanged because
the $d$-terms cancel around a cycle, so $\rho_{\text{trop}} = 0$ and
$\mathcal N$ remains in tropically normalised form.
\end{proof}

\begin{definition}[Stochastic regular expressions for the tropical semiring]
Instantiating \Cref{def:sre-syntax} with the closure solution $e=0$ from
the list after \Cref{thm:star-general}:
\begin{itemize}
    \item $\delta_w (w') = 0$ if $w=w'$, else $\infty$
    \item $(a_1 r_1 + a_2 r_2)(w) = \min(a_1 + r_1(w),\ a_2 + r_2(w))$ such that $\min(a_1,a_2) = 0$
    \item $(r_1 \cdot r_2)(w) = \min_{w = uv}\, \bigl[\, r_1(u) + r_2(v) \,\bigr]$
    \item $r^{*}_t(w) = \min_{k \geq 1} \min_{w = w_1 \cdots w_k} \left[ (k-1)t + \sum_{i=1}^{k} r(w_i) \right], \quad t \geq 0,\ w_i \in \Sigma^+$
\end{itemize}
All of these maintain the invariant that $\min_{w \in \Sigma^*} r(w) = 0$.
\end{definition}

The same results hold for the max-plus semiring and, as shown in
\Cref{app:viterbi}, for the Viterbi semiring.


\section{Conclusion and Future Work}
\label{sec:discussion}
We separated the question of when local stochasticity is a
well-posed, compositional invariant (Q1) from the question of when a
given automaton can be brought into locally stochastic form (Q2). Several directions remain open: 

\begin{enumerate}
\item \emph{Necessity of positivity.} \Cref{sec:rung2} shows that
completeness and positivity suffice for Q1. Are they necessary? Is
there a non-positive semiring over which the SRE grammar is
well-posed? Do matrix and polynomial semirings, such as
$\mathbb{R}_{\geq0}[x]$, satisfy the hypotheses?
\item \emph{Learning and testing.} Given sample strings from a rational
distribution $f$, can $\rho(M)$ be estimated efficiently? Can
$\varepsilon$-closeness of two distributions be tested given sampler
access?
\item \emph{A uniform realisability criterion.} Q2 is settled for
$\Rp$, the tropical semiring, and the Viterbi semiring, each by a
construction specific to that semiring. Is there a single sufficient
condition on the semiring that covers all three? The hypotheses differ.
\end{enumerate}
\clearpage

\bibliography{references}

@article{Ross2000,
  author  = {B. Ross},
  title   = {Probabilistic Pattern Matching and the Evolution of Stochastic Regular Expressions},
  journal = {Applied Intelligence},
  year    = {2000}
}

@mastersthesis{Beeh2017,
  author  = {T. Beeh},
  title   = {Transformations between Markov Chains and Stochastic Regular Expressions},
  school  = {University of Stuttgart},
  type    = {Bachelor Thesis},
  year    = {2017}
}

@inproceedings{agarwal2026transformers,
  author    = {Agarwal, Smayan and Faizi, Aslah Ahmad and Singh, Shobhit and Thakkar, Aalok},
  title     = {From Transformers to Weighted Automata: Towards the Verification of Large Language Models},
  booktitle = {From Data to Models and Back -- 13th International Symposium, {DataMod} 2025, Toledo, Spain},
  series    = {Lecture Notes in Computer Science},
  volume    = {16421},
  pages     = {145--162},
  publisher = {Springer},
  address   = {Cham},
  year      = {2026},
  doi       = {10.1007/978-3-032-25552-5_10}
}

@book{HornJohnson2012,
    author = {R. A. Horn and C. R. Johnson},
    title = {Matrix Analysis},
    publisher = {Cambridge University Press},
    address   = {Cambridge},
    year = {2012} 
}

@article{droste2007weighted,
  title     = {Weighted automata and weighted logics},
  author    = {Droste, Manfred and Gastin, Paul},
  journal   = {Theoretical Computer Science},
  volume    = {380},
  number    = {1-2},
  pages     = {69--86},
  year      = {2007},
  publisher = {Elsevier},
  doi       = {10.1016/j.tcs.2007.02.055}
}

@article{Denis2004,
  author    = {François Denis and Yann Esposito},
  title     = {On Rational Stochastic Languages},
  journal   = {Fundamenta Informaticae},
  volume    = {62},
  number    = {3-4},
  pages     = {197--213},
  year      = {2004},
  publisher = {IOS Press},
  issn      = {0169-2968}
}

@InProceedings{Denis2009LearningWA,
author="Denis, Fran{\c{c}}ois
and Esposito, Yann
and Habrard, Amaury",
editor="Lugosi, G{\'a}bor
and Simon, Hans Ulrich",
title="Learning Rational Stochastic Languages",
booktitle="Learning Theory",
year="2006",
publisher="Springer Berlin Heidelberg",
address="Berlin, Heidelberg",
pages="274--288",
}

@article{rabin1963probabilistic,
  title={Probabilistic Automata},
  author={Rabin, Michael O.},
  journal={Information and Control},
  volume={6},
  number={3},
  pages={230--245},
  year={1963},
  publisher={Elsevier}
}

@book{paz1971probabilistic,
  title={Introduction to Probabilistic Automata},
  author={Paz, Azaria},
  year={1971},
  publisher={Academic Press},
  address={New York}
}

@inproceedings{chatterjee2009probabilistic,
  title={Probabilistic Weighted Automata},
  author={Chatterjee, Krishnendu and Doyen, Laurent and Henzinger, Thomas A.},
  booktitle={Proceedings of CONCUR 2009},
  series={LNCS},
  pages={244--258},
  publisher={Springer},
  year={2009}
}

@article{chistikov2022bigo,
  author       = {Dmitry Chistikov and Stefan Kiefer and Andrzej S. Murawski and David Purser},
  title        = {The Big-O Problem},
  journal      = {Logical Methods in Computer Science},
  volume       = {18},
  number       = {1},
  pages        = {40:1--40:50},
  year         = {2022},
  doi          = {10.46298/LMCS-18(1:40)2022},
}

@inproceedings{chistikov2020bigo,
  author    = {Dmitry Chistikov and Stefan Kiefer and Andrzej S. Murawski and David Purser},
  title     = {The Big-{O} Problem for Labelled {M}arkov Chains and Weighted Automata},
  booktitle = {31st International Conference on Concurrency Theory (CONCUR 2020)},
  series    = {Leibniz International Proceedings in Informatics (LIPIcs)},
  volume    = {171},
  pages     = {41:1--41:20},
  publisher = {Schloss Dagstuhl--Leibniz-Zentrum f{\"u}r Informatik},
  year      = {2020},
  doi       = {10.4230/LIPIcs.CONCUR.2020.41}
}

@inproceedings{Bollig2012ProbabilisticKleene,
  author    = {Benedikt Bollig and Paul Gastin and Benjamin Monmege and Marc Zeitoun},
  title     = {A Probabilistic Kleene Theorem?},
  booktitle = {Proceedings of the 13th International Conference on Implementation and Application of Automata (CIAA 2012)},
  series    = {Lecture Notes in Computer Science},
  publisher = {Springer},
  year      = {2012},
  note      = {LSV, ENS Cachan, CNRS \& Inria, and LaBRI, Univ. Bordeaux \& CNRS},
  url       = {https://hal.science/hal-00713660/document}
}

@article{kleene1956,
  author    = {Stephen C. Kleene},
  title     = {Representation of events in nerve nets and finite automata},
  journal   = {Automata Studies},
  pages     = {3--42},
  year      = {1956},
  publisher = {Princeton University Press}
}

@article{vidal2005statistical,
  author    = {Enrique Vidal and Franck Thollard and Colin de la Higuera and 
               Francisco Casacuberta and Rafael C. Carrasco},
  title     = {Probabilistic finite-state machines---{Part I}},
  journal   = {IEEE Transactions on Pattern Analysis and Machine Intelligence},
  volume    = {27},
  number    = {7},
  pages     = {1013--1025},
  year      = {2005},
  doi       = {10.1109/TPAMI.2005.147}
}

@article{vanglabbeek1995reactive,
  title={Reactive, generative, and stratified models of probabilistic processes},
  author={van Glabbeek, Rob J and Smolka, Scott A and Steffen, Bernhard},
  journal={Information and Computation},
  volume={121},
  number={1},
  pages={59--80},
  year={1995},
  publisher={Elsevier}
}

@book{HMU06,
  author    = {John E. Hopcroft and Rajeev Motwani and Jeffrey D. Ullman},
  title     = {Introduction to Automata Theory, Languages, and Computation},
  edition   = {3rd},
  publisher = {Addison-Wesley},
  address   = {Boston, MA},
  year      = {2006}
}

@article{silva2011quantitative,
   author  = {Alexandra Silva and Filippo Bonchi and Marcello M. Bonsangue and Jan J. M. M. Rutten},
   title   = {Quantitative {K}leene coalgebras},
   journal = {Information and Computation},
   volume  = {209},
   number  = {5},
   pages   = {822--849},
   year    = {2011},
   doi     = {10.1016/j.ic.2010.09.007}
 }

@inproceedings{avni2015stochastization,
   author    = {Guy Avni and Orna Kupferman},
   title     = {Stochastization of Weighted Automata},
   booktitle = {Mathematical Foundations of Computer Science 2015 (MFCS 2015)},
   series    = {Lecture Notes in Computer Science},
   volume    = {9234},
   pages     = {89--102},
   publisher = {Springer},
   year      = {2015},
   doi       = {10.1007/978-3-662-48057-1_7}
 }

@book{eilenberg1974,
  author    = {Eilenberg, Samuel},
  title     = {Automata, Languages, and Machines, Vol.\ A},
  publisher = {Academic Press},
  year      = {1974}
}

\clearpage
\appendix
\section{Definitions and Preliminaries}
\label{app:matrix-properties}

\paragraph*{Weighted Automata Definitions}

\begin{definition}[Semiring]
A \emph{semiring} is a tuple $(S,+,\times,\SZero,\SOne)$ such that
$(S,+,\SZero)$ is a commutative monoid, $(S,\times,\SOne)$ is a monoid, $\times$ distributes over $+$ on both
sides ($a(b+c)=ab+ac$ and $(b+c)a=ba+ca$), and $\SZero$ is two-sided
annihilating: $\SZero\times a=a\times\SZero=\SZero$ for all $a\in S$.
\end{definition}

\begin{definition}[Complete semiring]
\label{def:complete}
$(S,+,\times,\SZero,\SOne)$ is \emph{complete} if $\sum_{i\in I}a_i$
is defined for every family $(a_i)_{i\in I}$ (finite or infinite), agrees with
the finite sum when $I$ is finite, is invariant under arbitrary partitioning
of $I$ (generalised associativity/commutativity of $+$), and
distributes over $\times$ on both sides: $a\sum_{i\in I}b_i=\sum_{i\in I}ab_i$
and $\big(\sum_{i\in I}b_i\big)a=\sum_{i\in I}b_ia$~\cite{eilenberg1974}.
\end{definition}

\begin{definition}[Mass]
\label{def:mass}
The \emph{mass} of $f:\Sigma^*\to S$ is $\|f\|=\sum_{w\in\Sigma^*}f(w)$,
defined because $S$ is complete.
\end{definition}

\begin{definition}[Positive / zero-sum-free semiring]
\label{def:positive}
$S$ is \emph{positive} (or \emph{zero-sum-free}) if
$\sum_{i\in I}a_i=\SZero$ implies $a_i=\SZero$ for every $i\in I$.
\end{definition}

\begin{definition}[Weighted finite automata, WFA]
\label{def:wfa}
A weighted automaton $A$ is a tuple $(\Sigma,Q,\lambda,M,\mu)$ over a
semiring $S$ where $\Sigma$ is a finite alphabet, $Q$ is a finite state set,
$\lambda:Q\to S$ is the initial weight vector, $M:\Sigma\to S^{Q\times Q}$
assigns a transition matrix $M_\sigma$ to each symbol, and $\mu:Q\to S$ is
the final weight vector. The semantics of $A$ is
$\llbracket A\rrbracket(w)=\lambda^\top M_{w_1}\cdots M_{w_n}\mu$ for
$w=w_1\cdots w_n\in\Sigma^*$; we abbreviate $f_A:=\llbracket A\rrbracket$.
\end{definition}

\begin{definition}[Trim automaton]
\label{def:trim}
A state is \emph{accessible} if it can be reached from some state of
nonzero initial weight along transitions of nonzero weight, and
\emph{co-accessible} if some state of nonzero final weight can be reached
from it. An automaton is \emph{trim} if every state is accessible and
co-accessible. Removing the states that fail either condition does not change the semantics.
\end{definition}

\paragraph*{Linear Algebra Results}

We recall standard results from nonnegative matrix theory
(see~\cite{HornJohnson2012}) that we use to analyse convergence of
weighted computations. The results here concern real square matrices in $\R^{n\times n}$.

\begin{enumerate}
\item A nonnegative square matrix $A$ is
\emph{irreducible} if for all $i,j$ there exists $k\ge1$ such that
$(A^k)_{ij}>0$.
\item The \emph{spectral radius} $\rho(\cdot)$ of a square matrix is the maximum of the absolute values of its eigenvalues. The spectral radius of a bounded operator is the supremum of the absolute values of the elements of its spectrum.
\item A nonnegative square matrix A is \emph{row sub-stochastic} if for every row $i$, $\sum_j A_{ij} \leq 1$. It is strictly sub-stochastic if for at least one row, the sum is strictly less than $1$.
\item (Perron-Frobenius Theory) If a nonnegative matrix $A$ is irreducible, then:
\[
\min_i \sum_j A_{ij} \le \rho(A) \le \max_i \sum_j A_{ij}.
\]
\item (Perron-Frobenius Theory) Let $A$ be an irreducible nonnegative matrix. Then ${(I-A)^{-1}}$ can be calculated if $\rho(A) < 1$.
\item (Neumann Series) For a square matrix $A$, the series $I + A + A^2 + A^3 \ldots $ converges if and only if $\rho(A) < 1$. In this case, it converges to $(I - A)^{-1}$.
\end{enumerate}

\section{Stochastic Regular Expressions}
\label{app:sre}

\begin{theorem}[Thompson Construction for SREs]
\label{thm:thompson-sre}
For every SRE $r$, there exists a probabilistic automaton $A_r$ with $O(|r|)$ states
such that $\llbracket A_r \rrbracket = \llbracket r \rrbracket$, where $|r|$ denotes
the size of $r$ (number of operators and symbols).
\end{theorem}

\begin{proof}
We construct $A_r$ by structural induction on $r$. Throughout, each automaton has
a unique initial state $q_0$ with $\lambda(q_0)=1$ and a unique final state $q_f$
with $\mu(q_f)=1$, and is locally stochastic, i.e.,
\[
\sum_{a \in \Sigma \cup \{\varepsilon\}} \sum_{q'} M_a(q,q') + \mu(q) = 1
\quad \text{for all states } q.
\]

\textbf{Base case.}
Let $r=\delta_\sigma$ for $\sigma \in \Sigma$.
Define $A_{\delta_\sigma}=(\Sigma,Q,\lambda,\{M_a\},\mu)$ with
$Q=\{q_0,q_f\}$, $\lambda(q_0)=1$, $\mu(q_f)=1$, and a single transition
$M_\sigma(q_0,q_f)=1$; all other weights are zero.
Local stochasticity holds trivially at both states, and the semantics coincide.
The atom $\delta_\varepsilon$ is realised the same way with an
$\varepsilon$-transition $M_\varepsilon(q_0,q_f)=1$, and $\delta_w$ for
$|w|\geq2$ factors as $\delta_{w_1}\cdots\delta_{w_n}$, handled by the
Cauchy product below.

\textbf{Inductive cases.}
Assume automata $A_1,A_2$ have already been constructed for $r_1,r_2$ on disjoint
state sets.

\begin{enumerate}
\item \textbf{Convex combination.}
Let $r=\alpha\cdot r_1 + (1-\alpha)\cdot r_2$ with $\alpha\in[0,1]$.
Introduce a fresh initial state $q_0$ and set
\[
M_\varepsilon(q_0,q_0^{(1)})=\alpha,
\qquad
M_\varepsilon(q_0,q_0^{(2)})=1-\alpha,
\]
where $q_0^{(i)}$ is the initial state of $A_i$.
The final state is obtained by identifying the final states of $A_1$ and $A_2$
into a single state $q_f$ with final weight $1$.
Local stochasticity holds at $q_0$ by construction, and elsewhere by induction.
The automaton realises the convex combination of the two behaviours.

\item \textbf{Cauchy product.}
Let $r=r_1\cdot r_2$.
Identify the unique final state $q_f^{(1)}$ of $A_1$ with the unique initial state
$q_0^{(2)}$ of $A_2$ by removing the final weight $\mu_1(q_f^{(1)})$ and replacing it
with an $\varepsilon$-transition of weight $1$.
The initial state is that of $A_1$, and the final state is that of $A_2$.
Local stochasticity is preserved, and the semantics coincide with concatenation.

\item \textbf{Discounted Kleene star.}
Let $r=(r_1)^*_t$, interpreted as a geometrically distributed iteration with
continuation weight $t\in[0,1)$ and at least one iteration.
Introduce a fresh initial state $q_0$ and final state $q_f$.
Add an $\varepsilon$-transition $M_\varepsilon(q_0,q_0^{(1)})=1$, where $q_0^{(1)}$
is the initial state of $A_1$.
Replace the final weight at $q_f^{(1)}$ by two $\varepsilon$-transitions:
\[
M_\varepsilon(q_f^{(1)},q_0^{(1)})=t,
\qquad
M_\varepsilon(q_f^{(1)},q_f)=1-t.
\]
Local stochasticity holds at $q_f^{(1)}$, and all other states are unchanged.
The automaton executes at least one iteration of $A_1$ and, after each
iteration, terminates with probability $1-t$. This gives the desired
geometric distribution.
\end{enumerate}

\textbf{Size bound.}
Each operator introduces at most a constant number of new states.
Thus $|A_r|=O(|r|)$.
\end{proof}

\begin{remark}[Removing $\varepsilon$-transitions]
\label{rem:eps-elimination}
The construction uses $\varepsilon$-transitions, which \Cref{def:wfa} does
not provide. They can be removed without adding states. Every starred
subexpression is guarded, so every $\varepsilon$-path from the initial to
the final state of a subautomaton has weight $0$; consequently the
positive-weight $\varepsilon$-transitions form an acyclic graph, and
$E := (I-M_\varepsilon)^{-1} = \sum_{n\ge0} M_\varepsilon^n$ is a finite
sum. The automaton $(\Sigma,Q,\lambda,\{E M_a\}_{a\in\Sigma},E\mu)$
computes the same function and is again locally stochastic:
$E\bigl(\sum_a M_a\mathbf 1 + \mu\bigr) = E(I-M_\varepsilon)\mathbf 1
= \mathbf 1$, using local stochasticity of $A_r$ in the form
$M_\varepsilon\mathbf 1 + \sum_a M_a\mathbf 1 + \mu = \mathbf 1$.
\end{remark}

\section{Conversions in the Viterbi semiring}
\label{app:viterbi}

The Viterbi semiring $([0,1],\max,\times,0,1)$ lies between the two
previous cases: its addition is idempotent, as in the tropical
semiring, and its multiplication is that of $\Rp$. Both features
make rescaling easier. Write $M=\max_{a\in\Sigma}M_a$ (entrywise) for the joint
transition matrix, and $\|f_A\|=\max_{w\in\Sigma^*}f_A(w)$ for the mass of a
Viterbi automaton $A$; a \emph{run} is a path from a state with nonzero
initial weight to a state with nonzero final weight, and its value is the
product of the initial weight, the transition weights, and the final
weight.

\begin{lemma}[Viterbi mass is finite and attained]
\label{lem:viterbi-mass}
Let $A$ be a trim Viterbi automaton. Then $\|f_A\|$ is attained by a run
that repeats no state, and $\|f_A\|\in(0,1]$.
\end{lemma}
\begin{proof}
Every run has value in $[0,1]$, a product of weights in $[0,1]$. If a run
repeats a state, deleting the cycle between the two occurrences divides
its value by the cycle's weight, which lies in $(0,1]$, so the value does
not decrease. The supremum over all runs therefore equals the maximum over
the finitely many runs that repeat no state. Trimness provides at least
one run of nonzero value.
\end{proof}

None of the hypotheses used for $\Rp$ or the tropical semiring is
needed here. Idempotence of $\max$ and weights bounded by $\SOne$ make every trim
Viterbi automaton finite-mass.

The \emph{best completion vector} is the analogue of the completion
vector $d=M^*\mu$ over $\Rp$ and of the least completion cost
(\Cref{thm:tropical-local}) over the tropical semiring.

\begin{definition}[Best completion vector]
\label{def:viterbi-d}
For a Viterbi automaton $A$ and state $q$, let
$d(q)=\max_{w\in\Sigma^*}\semantics{A_q}(w)$, where $A_q$ is $A$ with its
initial vector replaced by the vector that is $1$ at $q$ and $0$
elsewhere. Equivalently, $d=M^*\mu$, where $M^*=\max_{n\ge0}M^n$ is
computed in the Viterbi matrix semiring.
\end{definition}

\begin{lemma}[Bellman equation]
\label{lem:viterbi-bellman}
On a trim automaton, $d(q)\in(0,1]$ for every $q$, and
\[
d(q)\;=\;\max\Big(\max_{a\in\Sigma}\max_{q'\in Q}M_a(q,q')\,d(q'),\
\mu(q)\Big).
\]
The vector $d$ is computable in $O(|Q|^3\,|\Sigma|)$ time by value
iteration, or in $O(|E|+|Q|\log|Q|)$ time by Dijkstra's algorithm on the
weights $-\log M_a(q,q')$, which are nonnegative.
\end{lemma}
\begin{proof}
The bound $d(q)\in(0,1]$, and attainment of the maximum, follow from
\Cref{lem:viterbi-mass} applied to $A_q$. For the equation: a best run
from $q$ either stops at $q$, contributing $\mu(q)$, or takes some
transition with weight $M_a(q,q')$ and continues optimally from $q'$,
contributing $M_a(q,q')\,d(q')$; conversely, every such combination is the
value of some run from $q$. For computation, the iteration $d^{(0)}=\mu$,
$d^{(k+1)}=\max\!\big(Md^{(k)},\mu\big)$ has $d^{(k)}(q)$ equal to the
best value over runs from $q$ with at most $k$ transitions, and some
optimal run repeats no state, so the iteration stabilises after at most
$|Q|$ rounds of $O(|Q|^2|\Sigma|)$ operations each.
\end{proof}

\begin{theorem}[Viterbi decomposition]
\label{thm:viterbi-decomposition}
Let $A=(\Sigma,Q,\lambda,M,\mu)$ be a trim Viterbi automaton with mass
$Z=\|f_A\|$. Define
\[
\lambda'(q)=\frac{\lambda(q)\,d(q)}{Z},\qquad
M'_a(q,q')=\frac{M_a(q,q')\,d(q')}{d(q)},\qquad
\mu'(q)=\frac{\mu(q)}{d(q)}.
\]
Then $\mathcal N=(\Sigma,Q,\lambda',M',\mu')$ is a Viterbi automaton (all
weights lie in $[0,1]$), it is locally stochastic in the sense of
\Cref{def:local-stoch-general} (the sums there are maxima), and
\[
f_A(w)\;=\;Z\cdot f_{\mathcal N}(w)\qquad\text{for every }w\in\Sigma^*.
\]
The construction keeps the state set and runs in polynomial time.
\end{theorem}
\begin{proof}
\emph{Weights in $[0,1]$.} By \Cref{lem:viterbi-bellman},
$d(q)\ge M_a(q,q')\,d(q')$ and $d(q)\ge\mu(q)$, so $M'_a(q,q')\le1$ and
$\mu'(q)\le1$. For $\lambda'$: splitting a run into its first state and
the rest gives $\max_q\lambda(q)\,d(q)=\max_w f_A(w)=Z$, so
$\lambda'(q)\le1$. Trimness gives $d(q)>0$ and $Z>0$, so the divisions are
divisions of nonzero reals.

\emph{Local stochasticity.} Factoring $d(q)^{-1}$ out of the maximum,
\[
\max\Big(\max_{a,q'}M'_a(q,q'),\ \mu'(q)\Big)
=\frac{1}{d(q)}\max\Big(\max_{a,q'}M_a(q,q')\,d(q'),\ \mu(q)\Big)
=\frac{d(q)}{d(q)}=1
\]
by the Bellman equation, and $\max_q\lambda'(q)=Z/Z=1$ as computed above.

\emph{Semantics.} For a single run
$q_0\xrightarrow{w_1}q_1\cdots\xrightarrow{w_n}q_n$, the factors $d(q_i)$
telescope:
\[
\lambda'(q_0)\,\prod_{i=1}^{n}M'_{w_i}(q_{i-1},q_i)\;\mu'(q_n)
=\frac{1}{Z}\,\lambda(q_0)\,\prod_{i=1}^{n}M_{w_i}(q_{i-1},q_i)\;\mu(q_n).
\]
Taking the maximum over all runs labelled $w$ gives
$f_{\mathcal N}(w)=f_A(w)/Z$.
\end{proof}

Unlike \Cref{thm:spectral-decomposition,thm:tropical-decomposition},
this decomposition has no growth-rate component: the scalar $Z$ absorbs
the whole normalisation.

\begin{definition}[Stochastic regular expressions for the Viterbi semiring]
Instantiating \Cref{def:sre-syntax} with the closure solution $e=1$:
\begin{itemize}
  \item $\delta_w(w') = 1$ if $w = w'$, else $0$
  \item $(a_1 r_1 + a_2 r_2)(w) = \max(a_1\, r_1(w),\ a_2\, r_2(w))$ such that $\max(a_1,a_2) = 1$
  \item $(r_1 \cdot r_2)(w) = \max_{w=uv}\, r_1(u)\, r_2(v)$
  \item $r^*_t(w) = \max_{k\ge1} \max_{w=w_1\cdots w_k}\, t^{k-1} \prod_{i=1}^k r(w_i)$, \quad $t\in[0,1]$, $w_i\in\Sigma^+$
\end{itemize}
All of these maintain the invariant that $\max_{w\in\Sigma^*} r(w) = 1$.
\end{definition}

\begin{remark}[The operation beyond the axioms]
\label{rem:viterbi-division}
As over $\Rp$ and the tropical semiring, the construction uses an
operation that the semiring axioms do not supply, here division. The
quotients $M_a(q,q')\,d(q')/d(q)$ and $\mu(q)/d(q)$ stay inside
$[0,1]$ only because the Bellman equation bounds each numerator
by its denominator. The semiring $([0,1],\max,\times)$ is not a
semifield, since $1/x\notin[0,1]$ for $x<1$. This is the third instance
of the pattern behind Q2: division over $\Rp$, subtraction over the
tropical semiring, and division here, each given for the specific
semiring rather than derived from the axioms of
\Cref{sec:rung0,sec:rung2}.
\end{remark}
\end{document}